\documentclass[aps,pre,reprint,superscriptaddress,showpacs]{revtex4-1}

\usepackage[pdftex]{graphicx}
\setkeys{Gin}{clip=true,draft=false}
\DeclareGraphicsExtensions{.pdf}
\usepackage[cmex10]{amsmath}
\usepackage{amsthm}
\usepackage{amsfonts}
\usepackage[cspex,bbgreekl]{mathbbol}
\newcommand{\R}{{\ifmmode\mathbb{R}\else$\mathbb{R}$\fi}}
\newcommand{\K}{{\ifmmode\mathbb{K}\else$\mathbb{K}$\fi}}
\newcommand{\N}{{\ifmmode\mathbb{N}\else$\mathbb{N}$\fi}}
\newcommand{\C}{{\ifmmode\mathbb{C}\else$\mathbb{C}$\fi}}
\newcommand{\Z}{{\ifmmode\mathbb{Z}\else$\mathbb{Z}$\fi}}
\newcommand{\D}{\mathrm{d}}

\newcommand{\I}{\mathrm{i}}

\newcommand{\gei}{{\I}}

\newcommand{\fracd}[2]{\frac{\displaystyle #1}{\displaystyle #2}}
\newcommand{\partialder}[2]{\frac{\displaystyle\partial #1}{\displaystyle\partial #2}}
\newcommand{\totder}[2]{\frac{\displaystyle\D #1}{\displaystyle\D #2}}

\newcommand{\abs}[1]{{\left| #1 \right|}}
\newcommand{\norm}[1]{{\left\| #1 \right\|}}

\newcommand{\up}[1]{{\text{#1}}}

\newcommand{\bvec}{\mathbf{b}}
\newcommand{\qvec}{\mathbf{q}}
\newcommand{\fvec}{\mathbf{f}}

\newcommand{\zvec}{\mathbf{z}}
\newcommand{\tvec}{\mathbf{t}}

\newcommand{\Amat}{\mathbf{A}}

\newcommand{\Tmat}{\mathbf{T}}

\newcommand{\Fmat}{\mathbf{F}}
\newcommand{\Omegamat}{\mathbf{\Omega}}
\newcommand{\Gammamat}{\mathbf{\Gamma}}
\newcommand{\Phimat}{\mathbf{\Phi}}

\newcommand{\xvec}{\mathbf{x}}

\newcommand{\vvec}{\mathbf{v}}
\newcommand{\yvec}{\mathbf{y}}
\newcommand{\evec}{\mathbf{e}}

\newcommand{\avec}{\mathbf{a}}

\newcommand{\pvec}{\mathbf{p}}
\newcommand{\rvec}{\mathbf{r}}
\newcommand{\xivec}{\boldsymbol{\xi}}
\newcommand{\omegamat}{\boldsymbol{\omega}}

\newcommand{\Mmat}{\mathbf{M}}

\newcommand{\Kmat}{\mathbf{K}}

\newcommand{\esp}{\operatorname{\mathrm{e}}}
\newcommand{\espo}[1]{\esp^{\displaystyle #1}}
\newcommand{\dia}{\operatorname{\up{diag}}}
\newcommand{\diag}[1]{\dia\left\{\displaystyle #1\right\}}
\newcommand{\Real}{\operatorname{Re}}
\newcommand{\real}[1]{{\Real\left\{#1 \right\}}}
\newcommand{\Imag}{\operatorname{Im}}
\newcommand{\imag}[1]{{\Imag\left\{#1 \right\}}}
\newcommand{\alphavec}{\boldsymbol{\alpha}}
\newcommand{\betavec}{\boldsymbol{\beta}}

\newtheorem{lemma}{Lemma}
\newtheorem{theorem}{Theorem}

\begin{document}


\title{Application of Floquet theory to dynamical systems with memory}


\author{Fabio L. Traversa}
\email[]{fabiolorenzo.traversa@uab.es}
\affiliation{Departament d'Enginyeria Electr\`onica, Universitat Aut\`onoma de Barcelona, 08193-Bellaterra (Barcelona), Spain}
\affiliation{Department of Physics, University of California, San Diego, La Jolla, California 92093, USA}

\author{Massimiliano Di Ventra}
\email[]{diventra@physics.ucsd.edu}
\affiliation{Department of Physics, University of California, San Diego, La Jolla, California 92093, USA}

\author{Federica Cappelluti}
\email[]{federica.cappelluti@polito.it}
\affiliation{Dipartimento di Elettronica e Telecomunicazioni, Politecnico di Torino, 10129 Torino, Italy}

\author{Fabrizio Bonani}
\email[]{fabrizio.bonani@polito.it}
\affiliation{Dipartimento di Elettronica e Telecomunicazioni, Politecnico di Torino, 10129 Torino, Italy}


\date{\today}

\begin{abstract}
We extend the recently developed generalized Floquet theory [Phys. Rev. Lett. {\bf 110}, 170602 (2013)] to 
systems with infinite memory. In particular, we show that a lower asymptotic bound exists for the
Floquet exponents associated to such cases. As examples, we analyze the cases of an ideal 1D system, a Brownian particle, and a circuit resonator with an ideal transmission line. All these examples show the usefulness of this new approach
to the study of dynamical systems with memory, which are ubiquitous in science and technology.

\end{abstract}


\maketitle

\section{Introduction}

Floquet theory is a fundamental tool for expressing the solution of linear differential equations that have periodic coefficients \cite{Floquet,Farkas}. These types of equations are very common in several areas of science and technology, such as quantum \cite{Shirley,buttiker} and classical \cite{Gammaitoni,Haken} physics, chemistry \cite{Boland}, electronics \cite{Suarez,milanesi09,TMTT,Cappelluti2013}, noise analysis \cite{Demir1,noi,noiFN,IET}, and in general dynamical systems \cite{Guckenheimer}. In particular, it provides a versatile tool for the stability analysis of physical systems characterized by a periodic steady-state. However, till very recently Floquet theorem was limited to systems whose features depend instantaneously on time, i.e., described by memoryless equations.

On the other hand, systems with memory are by far more common than memoryless ones \cite{Max_review}, with an extremely wide range of applications \cite{Max_Proc_IEEE,TraversaNN}. This led some of us (FT, MD, and FB) to prove a generalization of Floquet theory, extended to a wide class of systems described by \textit{linear} memory operators \cite{PRL}. In that paper we have only provided the fundamental theorem of this generalized Floquet theory, and as a corollary we have proved Bloch's theorem for non-local (in space) potentials.

In the present contribution, we apply the general theorem proved in \cite{PRL} to the case of systems with linear memory, providing a formal extension of the theorem of Ref.~\onlinecite{PRL} to systems whose memory is infinite, and showing that a lower asymptotic bound exists for the Floquet exponents associated to such cases. Furthermore, we also provide a general numerical tool based on the harmonic balance method \cite{Sangiovanni,IJCTA} aimed at the numerical assessment of stability. As practical examples, we discuss the analysis of an ideal one-dimensional dynamical system with memory, of a more
realistic Brownian particle with memory, and of a circuit resonator with ideal transmission line.

We consider here a nonlinear dynamical system expressed as
\begin{equation}
\totder{\zvec}{t}=\fvec(\zvec(t),t)+\int_{-\infty}^t \Kmat(t,\tau)\xvec(\tau)~\D\tau,
\label{xu}
\end{equation}
where $\xvec,\fvec\in\mathbb{R}^n$ and $\Kmat\in\mathbb{R}^{n\times n}$ is the kernel representing the memory effects. We assume that \eqref{xu} admits a solution $\zvec_\up{S}(t)$ periodic of period $T$ (limit cycle), and that $\Kmat$  is such that
\begin{subequations}
\label{cond}
\begin{align}
&\Kmat(t,\tau)=\Kmat(t+T,\tau+T) \label{cond1}\\
&\int_{-\infty}^t \norm{\Kmat(t,\tau)}~\D\tau<\infty\qquad \forall t
\label{cond2}
\end{align}
\end{subequations}
where $\norm{\cdot}$ is a properly defined norm.
Notice that if the memory part is time-invariant, i.e., if $\Kmat(t,\tau)=\Kmat(t-\tau)$, condition \eqref{cond1} is always satified.

The stability of the limit cycle $\zvec_\up{s}(t)$ is defined by the variational problem
\begin{equation}
\totder{\yvec}{t}=\Amat(t)\yvec(t)+\int_{-\infty}^t \Kmat(t,\tau)\yvec(\tau)~\D\tau,
\label{gu}
\end{equation}
where $\yvec(t)=\zvec(t)-\zvec_\up{S}(t)$ is the cycle perturbation, and $\Amat(t)$ is the $T$-periodic Jacobian matrix of $\fvec(\zvec(t),t)$, with respect to $\zvec$, calculated in the limit cycle.

\section{Generalized Floquet theorem}

The generalization of Floquet theorem proved in \cite{PRL} shows that the state transition matrix of \eqref{gu} can be expressed as
\begin{equation}
\Phimat(t;t_0)=\Mmat(t;t_0)\espo{\Fmat(t-t_0)}
\end{equation}
where $\Mmat(t;t_0)$ is $T$-periodic with respect to both time variables, and $\Fmat$ is a constant matrix whose eigenvalues constitute the cycle Floquet exponents. With respect to a memoryless systems, however, the size $p$ of $\Fmat$ may be larger than $n$, and even infinite.

Thus, the general solution of \eqref{gu} can be expressed as a linear combination of $p$ exponential functions, characterized by the Floquet exponent $\lambda$, times a $T$ periodic function $\rvec(t)=\rvec(t+T)$, the (direct) Floquet eigenvector. This means that we seek solutions of \eqref{gu} in the form
\begin{equation}
\yvec(t)=\rvec(t)\espo{\lambda t}\qquad \rvec(t)=\rvec(t+T).
\label{gensol}
\end{equation}
Due to the $T$-periodicity of the Floquet eigenvector, we have actually a set of $p$ different \textit{classes} of values for $\lambda$. In fact if $\lambda_0$ is one of the eigenvalues of $\Fmat$, each $\lambda_0+k\gei 2\pi/T$ ($k\in\mathbb{Z}$) spans the same eigenspace: we call this phenomenon the \textit{splitting} of the eigenvalues. To reduce to a minimum the number of significant quantities, it is customary to define the Floquet multipliers $\mu=\exp(\lambda T)$, since the exponential function eliminates the splitting phenomenon. Of course, the stability of the solution of \eqref{xu} depends on the sign of the real part of $\lambda$, or equivalently on the magnitude of $\mu$.

The defining equation for the Floquet eigenvalues (and direct eigenvectors) can be found substituting \eqref{gensol} into \eqref{gu}. We find
\begin{equation}
\totder{\rvec}{t}+\lambda\rvec(t)=\Amat(t)\rvec(t)+\qvec(\rvec,t,\lambda)
\label{FEdeft}
\end{equation}
where
\begin{equation}
\qvec(\rvec,t,\lambda)=\int_{-\infty}^t \Kmat(t,\tau)\rvec(\tau)\espo{\lambda(\tau-t)}~\D\tau\label{memory_operator}
\end{equation}
is a $T$-periodic function of $t$ because all the other terms of \eqref{FEdeft} are.

The convergence of the integral defining $\qvec(\rvec,t,\lambda)$ is not trivially derived from \eqref{cond2}. In order to clarify the matter, we start by proving the following Lemma:
\begin{lemma}\label{lemma1}
Let us consider \eqref{gu} where the memory kernel satisfies \eqref{cond}. We consider a real $\bar{s}$ and any $s>\bar{s}$. Then, the solutions ${\yvec}(t)={\rvec}(t)\exp({\lambda} t)$ and $\bar{\yvec}(t)=\bar{\rvec}(t)\exp(\bar{\lambda} t)$ of the variational problems (with finite memory)
\begin{subequations}
\label{finitemem}
\begin{align}
&\totder{\rvec}{t}+\lambda\rvec(t)=\Amat(t)\rvec(t)+\int_{t-s}^t \Kmat(t,\tau)\rvec(\tau)\espo{\lambda(\tau-t)}~\D\tau \label{finitemem1}\\
&\totder{\bar\rvec}{t}+\bar\lambda\bar\rvec(t)=\Amat(t)\bar\rvec(t)+\int_{t-\bar s}^t \Kmat(t,\tau)\bar\rvec(\tau)\espo{\bar\lambda(\tau-t)}~\D\tau
\label{finitemem2}
\end{align}
\end{subequations}
satisfy:
\begin{subequations}
\label{conditions}
\begin{equation}
\abs{\lambda-\bar\lambda}\le M \int_{t-s}^{t-\bar s}\max_{t\in[0,T]}\norm{\Kmat(t,\tau)}~\D\tau \label{conditions1}
\end{equation}
and for $\bar s\gg 0$
\begin{align}
& \abs{\lambda-\bar\lambda}\le M_- \int_{t-s}^{t-\bar s}\fracd{\max_{t\in[0,T]}\norm{\Kmat(t,\tau)}}{t-\tau}~\D\tau \nonumber\\
&\qquad\qquad \qquad\qquad \qquad \qquad \qquad  \up{if $\real{\bar\lambda}<0$}\label{conditions2}\\
&
\abs{\lambda-\bar\lambda}\le M_+ \int_{t-s}^{t-\bar s}\max_{t\in[0,T]}\norm{\Kmat(t,\tau)}\espo{-\real{\bar\lambda}(t-\tau)}~\D\tau  \nonumber\\
&\qquad\qquad \qquad\qquad \qquad \qquad \qquad  \up{if $\real{\bar\lambda}>0$}\label{conditions3}
\end{align}
\end{subequations}
where $0<M_+,M_-\le M< +\infty$.
\end{lemma}
\begin{proof}
See Appendix~\ref{Lemmaproof}.
\end{proof}
Lemma~\ref{lemma1}, starting from \eqref{cond2}, shows that for $\bar s$ large enough, for every $s$, $\abs{\lambda-\bar\lambda}\rightarrow 0$ faster than $\int_{t-s}^{t-\bar s}\norm{\Kmat(t,\tau)}~\D\tau$. This means that the eigenvalue $\lambda$ becomes independent of $\bar s$, i.e., of any finite approximation of the system memory length. The first consequence is numerical: for large enough $\bar s$, the eigenvalues of an infinite memory system can be calculated with a prescribed accuracy. Second, from a theoretical standpoint we have
\begin{theorem}\label{theo1}
The Floquet exponents $\lambda$ of \eqref{gu} satisfy
\begin{equation}
\real{\lambda}>-\min_{t\in[0,T]} k_\up{c}(t)
\label{simba1}
\end{equation}
where the critical exponent $k_\up{c}(t)$ is defined as
\begin{equation}
k_\up{c}(t)=\lim_{\tau\rightarrow -\infty} \fracd{\ln\norm{\Kmat(t,\tau)}}{\tau}.
\label{tembo}
\end{equation}
\end{theorem}
\begin{proof}
See Appendix~\ref{theoproof}.
\end{proof}

\section{Floquet exponents computation}
\label{FEcomp}

Clearly, \eqref{FEdeft} represents a generalized eigenvalue problem whose solution provides the required Floquet quantities for the limit cycle. The explicit expression for such a generalized eigenvalue problem depends on the features of the memory kernel $\Kmat(t,\tau)$. Since all the terms of \eqref{FEdeft} are $T$-periodic, a viable solution strategy is the use of frequency-domain approaches such as the Harmonic Balance (HB) technique \cite{Sangiovanni,IJCTA,AEU,TCAD}, here summarized in Appendix~\ref{AHB}. Frequency transformation of \eqref{FEdeft} yields
\begin{equation}
\Omegamat\tilde\rvec+\lambda\tilde\rvec=\tilde\Amat\tilde\rvec+\tilde\qvec(\tilde\rvec,\omega,\lambda)
\label{FEdeff}
\end{equation}
where $\omega$ is the set of all the frequencies multiple of the fundamental one $\omega_0=2\pi/T$. Eq.~\eqref{FEdeff} is a generalized, transcendental eigenvalue problem in $\lambda$ and $\tilde\rvec$, the collection of the harmonic amplitudes of the Floquet eigenvector $\rvec(t)$.

The explicit form of \eqref{FEdeff} depends on the type of $\Kmat(t,\tau)$. However, in general it can be transformed, at least approximately, into a polynomial eigenvalue problem by formally developing $\tilde\qvec(\tilde\rvec;\lambda)$ into Taylor series as a function of $\lambda$
\begin{equation}
\Omegamat\tilde\rvec+\lambda\tilde\rvec=\tilde\Amat\tilde\rvec+\sum_{k=0}^{+\infty}   \left.\partialder{\tilde\qvec(\tilde\rvec,\omega,\lambda)}{\lambda}\right|_{\lambda=0} \fracd{\lambda^k}{k!}.
\label{FEdeffseries}
\end{equation}
Several techniques are available to tackle the polynomial eigenvalue problems, such as for instance those discussed in \cite{Lanczos,Tisseur,Higham}.

Notice that \eqref{FEdeff} is an exact representation of the generalized, time-domain eigenvalue problem \eqref{FEdeft} only if  infinite Fourier series  are considered. Clearly, for practical calculations the series is truncated to a finite number of harmonics $N_\up{H}$ (see Appendix~\ref{AHB}) and, equivalently, the time domain problem is time-sampled. The truncation affects the accuracy of the Floquet quantities, especially on the exponents, as discussed e.g. in \cite{IJCTA,IET}. However, according to intuition, accurate results can be obtained by properly choosing $N_\up{H}$.

\section{Examples}

\subsection{A simple 1D dynamical system with memory}

The first example is an extremely simple dynamical system with memory, characterized by the 1D variational equation
\begin{equation}
\totder{y}{t}=ay(t)+\int_{t-s}^t \espo{k(\tau-t)}y(\tau)~\D\tau,
\label{pombe}
\end{equation}
where $a$ and $k>0$ are real parameters. Using the harmonic balance approach discussed in Sec.~\ref{FEcomp}, we find the explicit form of \eqref{FEdeff}:
\begin{equation}
(\lambda+\gei\omega_j)r_j=ar_j+\fracd{1-\espo{-s(k+\lambda+\gei\omega_j)}}{k+\lambda+\gei\omega_j}r_j.
\label{ulanzi}
\end{equation}
Since \eqref{pombe} is scalar, we can simplify $r_j$ from \eqref{ulanzi}, and consider the case $\omega_j=0$ since the roots of the general eigenvalue equation are simply those for $\omega_j=0$ plus a shift equal to $-\gei\omega_j$. In other words, we have to study the transcendental eigenvalue equation
\begin{equation}
\lambda=a+\fracd{1-\espo{-s(k+\lambda)}}{k+\lambda}.
\label{ulanzi1}
\end{equation}

\begin{figure}
\centerline{
\includegraphics[width=.866\columnwidth]{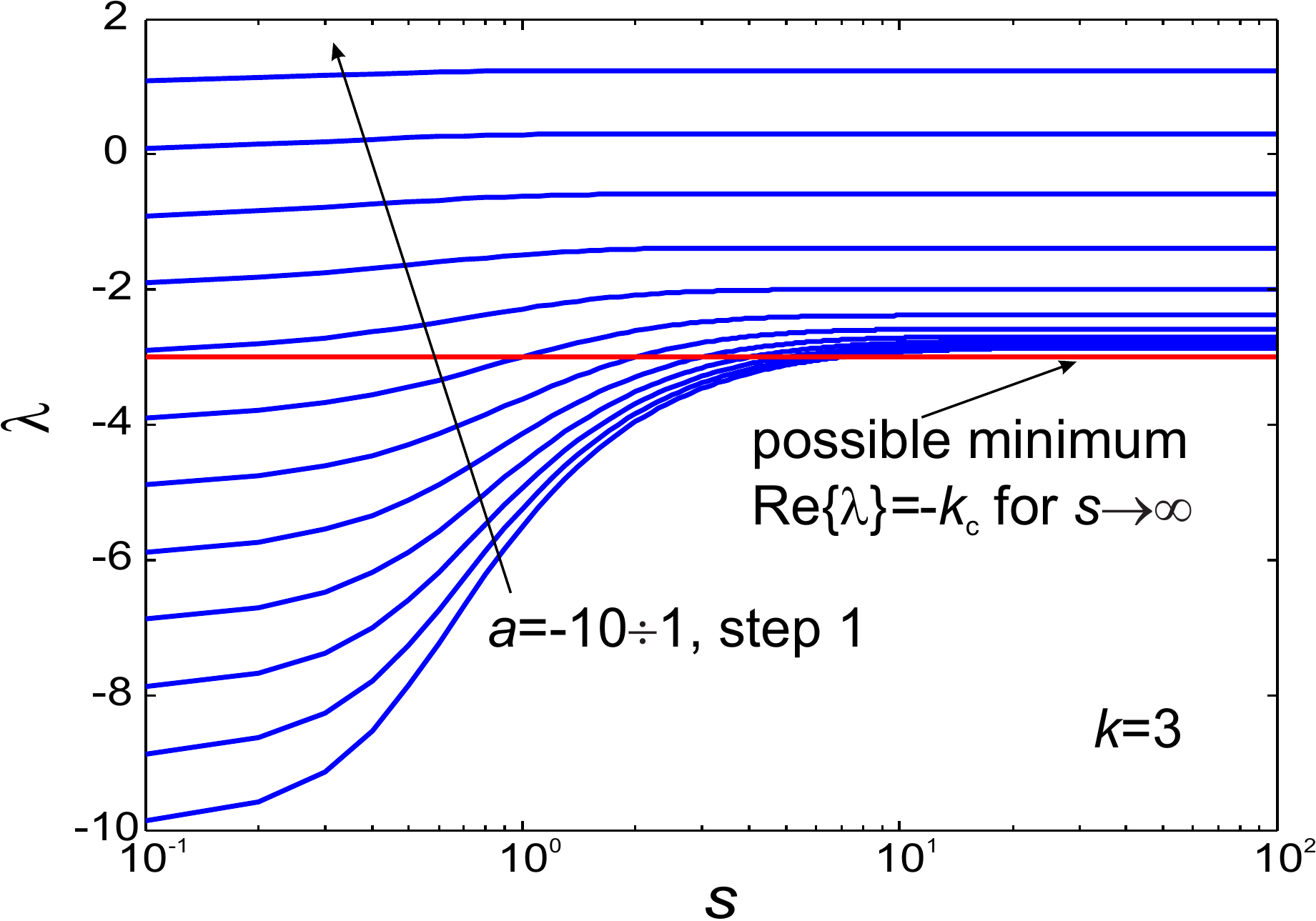}}
\caption{\label{1D1}(Color online) 1D system with memory: Dependence on $s$ of the solution of \eqref{ulanzi1} as a function of the model paramaters. The horizontal line represents the asymptotic possible minimum of $\real{\lambda}$ according to Theorem~\ref{theo1}.}
\end{figure}

Equation \eqref{ulanzi1} has only one solution, whose dependence on the $s$ and $a$ parameters ($k$ is kept fixed at 3) is sketched in Fig.~\ref{1D1}. As expected from Lemma~\ref{lemma1}, for each value of the $a$ parameter, the solution $\lambda$ reaches an asymptotic value for ``long'' enough memory, i.e., for $s$ large enough. On the other hand, since the memory kernel is exponential, the critical exponent is easily calculated as $k_\up{c}=k$, thus setting the lower bound for the real part of $\lambda$ (see Fig.~\ref{1D1}).

\begin{figure}
\centerline{
\includegraphics[width=.866\columnwidth]{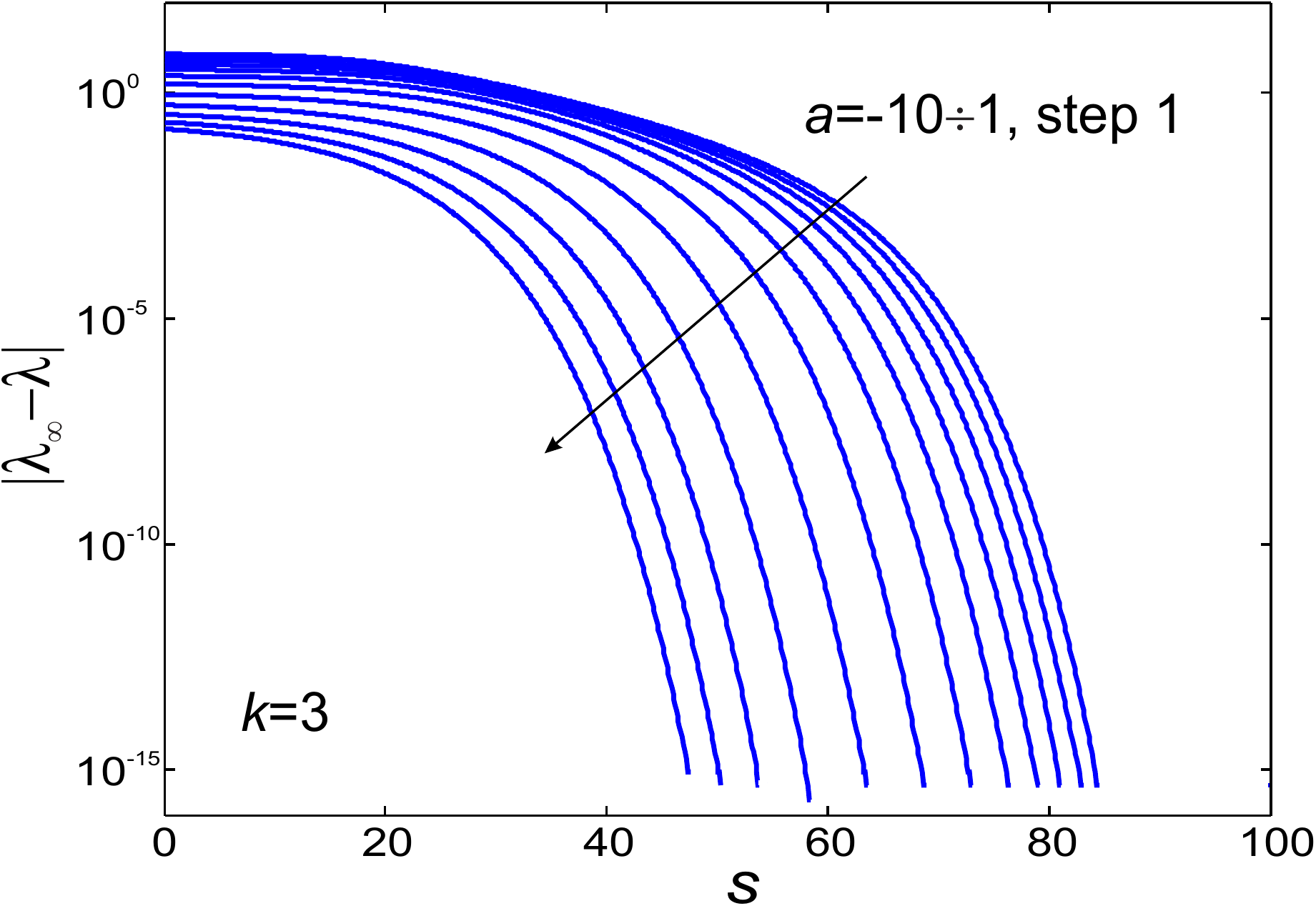}}
\caption{\label{1D2}(Color online) 1D system with memory: Evolution as a function of $s$ of the difference between the computed $\lambda$ and its asymptotic value for $s\rightarrow+\infty$.}
\end{figure}

In order to study the dependence of $\lambda$ on the memory parameter, we show in Fig.~\ref{1D2} the absolute difference with respect to the asymptotic solution $\lambda_\infty$ as a function of $s$ for several $a$ values (and for $k=3$). The solution converges very rapidly to $\lambda_\infty$ as $s$ grows above a threshold, whose value mildly depends on $a$.

\subsection{Brownian particle with memory}

The second example of application that we consider is the 2D Brownian particle with memory introduced in \cite{ICNF2013mem}, that in turn stems from the model studied in \cite{Erdmann2002}.

According to \cite{Kubo}, we introduce memory by exploiting a friction retardation function $\Gamma(t,\tau)$ that modifies the 2D particle equations of motion into \cite{ICNF2013mem} (where $m$ is the particle mass)
\begin{subequations}
\label{mem}
\begin{align}
\totder{\xvec_\up{p}}{t}&=\vvec_\up{p}(t) \\[1ex]
m\totder{\vvec_\up{p}}{t}&=-m\int_{-\infty}^t \Gamma(t,\tau)\vvec_\up{p}(\tau)~\D\tau-\nabla U(\xvec_\up{p})+\xivec(t).\label{mem2}
\end{align}
\end{subequations}
The external potential $U$ represents a central force as in \cite{Erdmann2002}
\begin{equation}
\nabla U(\xvec_\up{p})=m\bar\omegamat^2\xvec_\up{p},
\end{equation}
where $\bar\omegamat=\diag{\bar\omega_1,\bar\omega_2}$ and $\bar\omega_1,\bar\omega_2>0$ are two real parameters.

Assuming a colored Langevin force $\xivec(t)$ described by an Ornstein-Uhlenbeck process, the friction retardation function reads \cite{Hanggi1995}
\begin{equation}
\Gamma(t,\tau)=\gamma k\espo{-k\abs{t-\tau}},
\end{equation}
where $\gamma$ is the friction function, instantaneously dependent on the particle velocity \cite{Kubo}, and $k=1/\tau_\up{n}$, $\tau_\up{n}$ being the noise correlation time.

Taking the ensemble average provides the deterministic system
\begin{subequations}
\label{memnoL}
\begin{align}
\totder{\xvec}{t}&=\vvec(t) \\[1ex]
m\totder{\vvec}{t}&=-m\int_{-\infty}^t \Gamma(t,\tau)\vvec(\tau)~\D\tau-\nabla U(\xvec)
\end{align}
\end{subequations}
where we assume a first order (as a function of $\tau_\up{n}=1/k$) correction to the Rayleigh model \cite{ICNF2013mem}:
\begin{equation}
\gamma[v(\tau)]=-\alpha+\beta v(\tau)^2+\fracd{g}{k}.
\end{equation}
Coefficients $\alpha, \beta, g$ and $k$ are model parameters, while $v$ is the magnitude of the particle velocity $\vvec$. As discussed in \cite{ICNF2013mem}, the case of the memoryless particle is obtained by letting $k\rightarrow+\infty$.

System \eqref{memnoL} admits a periodic limit cycle in the phase space $(\xvec_\up{S}(t),\vvec_\up{S}(t))$ for several values of the parameters. The corresponding stability is assessed following the procedure outlined in Sec.~\ref{FEcomp}. The 4D Floquet eigenvector $\rvec(t)$ is decomposed in the 2D position $\rvec_x(t)$ and velocity $\rvec_v(t)$ components. Substituting into \eqref{FEdeft} we evaluate the integral in \eqref{memory_operator} taking into account the Fourier expansion of $\rvec(t)$ and  Theorem \ref{theo1}, that guarantees $\real{\lambda}>-k$. Thus we get a closed, albeit in infinite series form, expression
\begin{subequations}
\label{memLF}
\begin{align}
\totder{\rvec_x}{t}+\lambda\rvec_x&=\rvec_v \\[1ex]
\totder{\rvec_v}{t}+\lambda \rvec_v&=
-k\sum_{j,h}  \fracd{\tilde\gamma_{\up{e},j-h}\tilde{\rvec}_{v,h}\espo{ \gei\omega_{j}t}}{k+\lambda+\gei\omega_{j}}-\bar{\omegamat}^2\rvec_x
\end{align}
\end{subequations}
where $\tilde{\rvec}_{\alpha,j}$ is the $j$-th harmonic amplitude (in exponential form) for function $\rvec_\alpha(t)$ ($\alpha=x,v$), $j,h\in\mathbb{Z}$ and $\omega_j=j2\pi/T$.

Equation \eqref{memLF} is further transformed into \eqref{FEdeff} exploiting the Fourier series expansion of all the $T$-periodic terms. Balancing the harmonic components, we find the ideally infinite set (as a function of the harmonic index $j$)  of  order 2 polynomial eigenvalue problems (PEP)
\begin{subequations}
\label{memLHBPEP}
\begin{align}
(\lambda+\gei\omega_j)\tilde{\rvec}_{x,j}&=\tilde{\rvec}_{v,j}\\[1ex]
(\lambda+\gei\omega_j)(k+\lambda+\gei\omega_{j})\tilde{\rvec}_{v,j}&=- (k+\lambda+\gei\omega_{j})\bar{\omegamat}^2\tilde{\rvec}_{x,j} \nonumber\\
&\phantom{=}~-k \sum_h \tilde\gamma_{\up{e},j-h}\tilde{\rvec}_{v,h}.
\end{align}
\end{subequations}
As discussed in \cite{Tisseur,Higham}, an $r$-th order PEP for an equation of size $m$ has $r\times m$ eigenvalues, thus \eqref{memLHBPEP} provides a set of $n/2+2(n/2)=3n/2$ eigenvalues (for each harmonic $j$), where for the 4D phase space of the 2D particle $n=4$, i.e., a total of 6 classes of Floquet multipliers.

\begin{figure}
\centerline{
\includegraphics[width=.866\columnwidth]{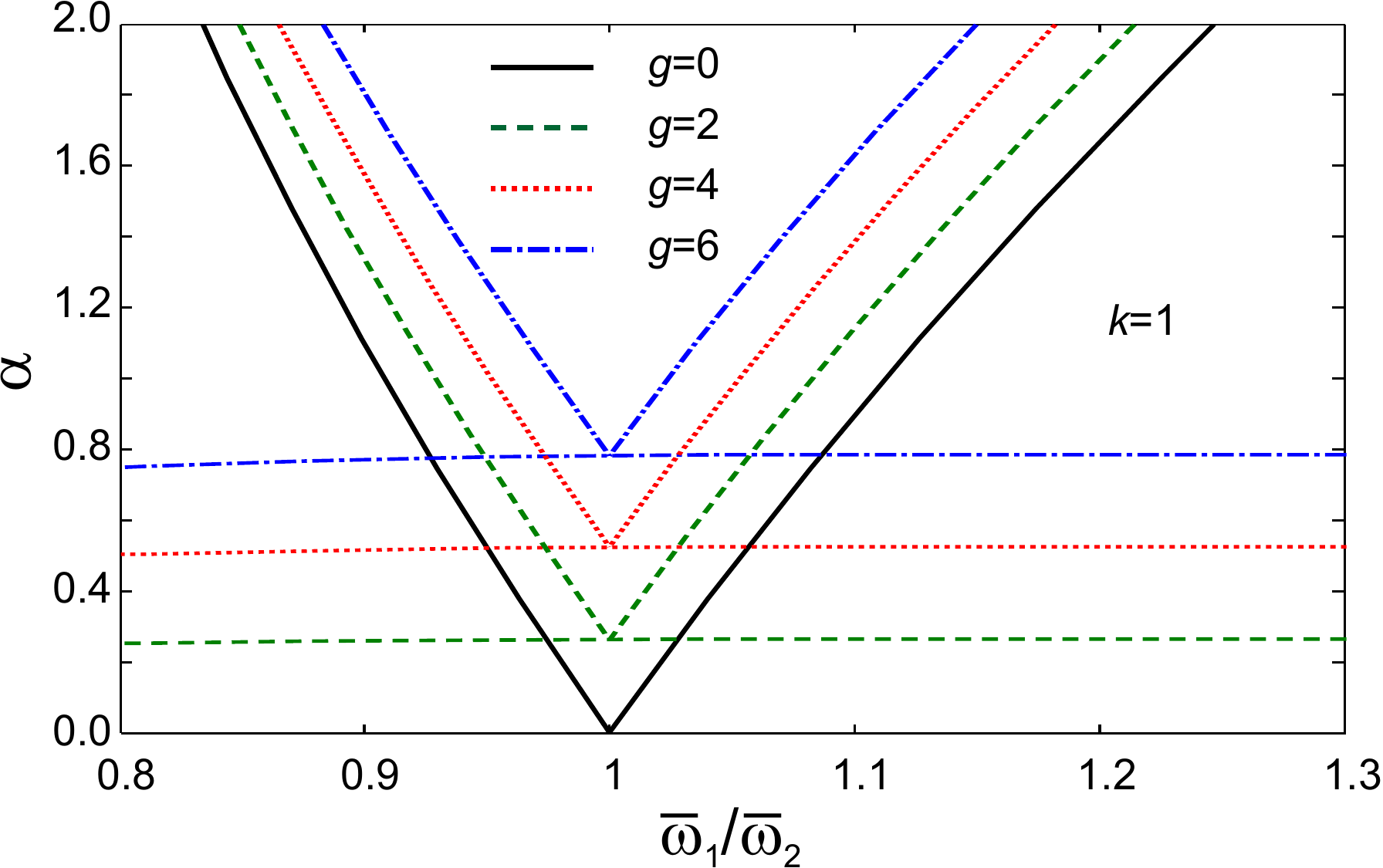}}
\caption{\label{sol} (Color online) Brownian particle with memory: Arnold tongue in the $(\bar\omega_1/\bar\omega_2,\alpha)$ plane as a function of $g$. Parameters: $\beta=1$, $\bar\omega_1=2$, $k=1$.}
\end{figure}

We implemented a numerical solution of the second-order polynomial eigenvalue problem using $N_\up{H}=30$ harmonics, thus truncating \eqref{memLHBPEP} into a system for $j=-N_\up{H},\dots,N_\up{H}$. Furthermore, the limit cycle solution is calculated in the frequency domain exploiting the Harmonic Balance \cite{IJCTA} numerical technique, again making use of 30 harmonics.

With respect to the memoryless case (i.e., $k\rightarrow +\infty$) a stable equilibrium (originally unstable) is found in the origin of the phase space. The bifurcation diagram in the parameter space $(\bar{\omega}_1/\bar{\omega}_2,\alpha)$ is plotted in Fig.~\ref{sol} for $k=1$ and several values of $g$. As expected, Arnold tongues are found. The area of the parameter space below the almost horizontal curve (not present in \cite{Erdmann2002}) corresponds to the stable equilibrium previously discussed.  Above this line, we have either two symmetric, stable limit cicles (the inner part of the Arnold tongue) or a strange attractor (the outer part): both of these have a shape close to that of the memoryless case \cite{Erdmann2002}. The boundaries of the Arnold tongue correspond to fold bifurcations, thus implying that (as for the memoryless case) the strange attractor is generated by the collapse of the two limit cycles. The crossing point between the Arnold tongue and the boundary of the stable equilibrium defines a \textit{triple} point \footnote{The point is triple in the sense that a limit cycle, a stable equilibrium and a strange attractor co-exist for this combination of parameter values.} for the particle dynamics.

\begin{figure}
\centerline{
\includegraphics[width=.9\columnwidth]{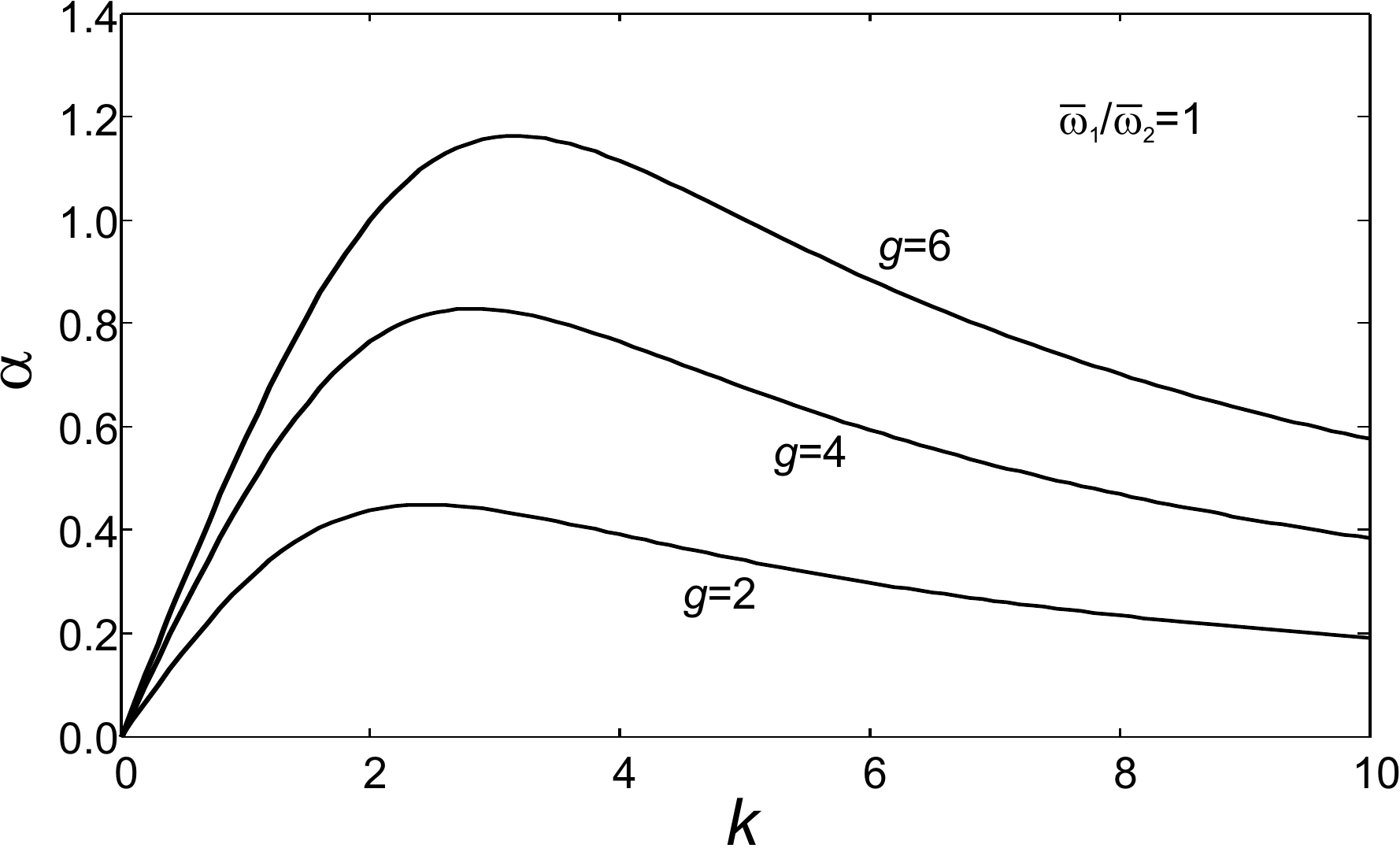}}
\caption{\label{alphakappa}(Color online) Brownian particle with memory: bifurcation curve in the $(k,\alpha)$ plane as a function of $g$ for the triple point. Parameters: $\beta=1$, $\bar\omega_1/\bar\omega_2=1$, $\bar\omega_1=2$.}
\end{figure}

Fig.~\ref{alphakappa} represents the bifurcation curve in the $(k,\alpha)$ plane for the triple point as a function of $g$. Consistently, the bifurcation curves tend to zero for disappearing memory ($k\rightarrow +\infty$).

\subsection{Circuit resonator with ideal transmission line}

\begin{figure}
\centerline{
\includegraphics[width=.4\columnwidth]{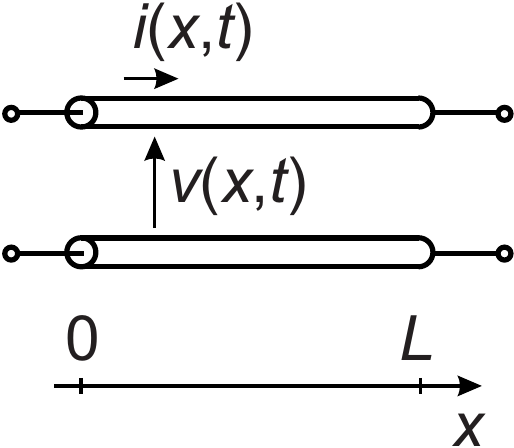}}
\caption{\label{tl}(Color online) Circuit representation of a transmission line (TL).}
\end{figure}

The last example we provide considers the presence of a lossless transmission line (TL) in an electronic circuit. TLs are ubiquitous in circuits for high frequency applications, e.g., in the RF and microwave frequency range \cite{Sorrentino}. The lossless TL is a distributed circuit element characterized by the following linear partial differential system of equations
\begin{subequations}
\label{lineeq}
\begin{align}
\partialder{v}{x}&=-\mathit{l}\partialder{i}{t} \\
\partialder{i}{x}&=-\mathit{c}\partialder{v}{t}
\end{align}
\end{subequations}
where $v(x,t)$ and $i(x,t)$ are, respectively, the voltage and current at time $t$ and position $x$ along the TL (see Fig.~\ref{tl}), $l$ is the TL inductance per unit length, and $c$ the TL capacitance per unit length. The general solution of \eqref{lineeq} is expressed as the sum of a progressive and a regressive wave \cite{Sorrentino}
\begin{subequations}
\label{lineeqsol}
\begin{align}
v(x,t)&=v_-(t-x/v_\up{f})+v_+(t+x/v_\up{f}) \\
i(x,t)&=Y_0v_-(t-x/v_\up{f})-Y_0v_+(t+x/v_\up{f})
\end{align}
\end{subequations}
where $v_\up{f}=1/\sqrt{\mathit{l}\mathit{c}}$ is the TL phase velocity and $Y_0=1/Z_0=\sqrt{\mathit{c}/\mathit{l}}$ is the TL characteristic admittance. The specific shape of $v_-$ and $v_+$ depends on the line boundary and initial conditions, i.e., on the circuit in which the TL is embedded. If we consider the TL embedded into a nonlinear circuit characterized by a set of state variables collectively denoted as $\yvec(t)$, the circuit equations to be solved take the form
\begin{subequations}
\label{lineecir}
\begin{align}
&\mathcal{L}_0\left\{v(0,t),i(0,t),\yvec,\dot\yvec\right\}=0 \\
&\mathcal{L}_L\left\{v(L,t),i(L,t),\yvec,\dot\yvec\right\}=0 \\
&\totder{\yvec}{t}=\qvec(\yvec,t)
\label{lineecir3}
\end{align}
\end{subequations}
where $\dot\yvec=\D\yvec/\D t$, $\mathcal{L}_0$ and $\mathcal{L}_L$ are linear operators, $\qvec$ is a vector function describing the embedding circuit state equations. System \eqref{lineecir} can be easily transformed into a nonlinear system with memory of type \eqref{xu}. In fact, from \eqref{lineeqsol}
\begin{subequations}
\label{lineeqsolporte}
\begin{align}
v(0,t)&=v_-(t)+v_+(t) \\
i(0,t)&=Y_0v_-(t)-Y_0v_+(t) \\
v(L,t)&=\int_{-\infty}^t v_-(\tau)\delta(t-\tau_\up{f}-\tau)\D\tau\nonumber\\
&\qquad+\int_t^{\infty} v_+(\tau)\delta(t+\tau_\up{f}-\tau)\D\tau \label{lineeqsolportec}\\
i(L,t)&=Y_0\int_{-\infty}^t v_-(\tau)\delta(t-\tau_\up{f}-\tau)\D\tau\nonumber\\
&\qquad-Y_0\int_t^{\infty} v_+(\tau)\delta(t+\tau_\up{f}-\tau)\D\tau \label{lineeqsolported}
\end{align}
\end{subequations}
where $\tau_\up{f}=L/v_\up{f}$ represents the delay associated to the TL. Notice that the time anticipation included in the second integral of \eqref{lineeqsolportec} and of  \eqref{lineeqsolported} is only apparent \cite[Chapter 8]{Orta}. In fact, reversing \eqref{lineeqsol} we find
\begin{subequations}
\label{lineeqsolrev}
\begin{align}
v_-(t-x/v_\up{f})&=\fracd{1}{2}v(x,t)+\fracd{1}{2Y_0}i(x,t)\\
v_+(t+x/v_\up{f})&=\fracd{1}{2}v(x,t)-\fracd{1}{2Y_0}i(x,t).
\end{align}
\end{subequations}

The variational problem defining the stability of the solution of \eqref{lineecir} can thus be derived by linearizing the full system and looking for solutions of type \eqref{gensol}, i.e., $\delta v(x,t)=r_v(x,t)\exp(\lambda t)$ and $\delta i(x,t)=r_i(x,t)\exp(\lambda t)$. Correspondingly, from \eqref{lineeqsol} we find $r_{v_-}(t-x/v_\up{f})=r_{v_-}(t)\exp[\lambda(t-x/v_\up{f})]$ and $r_{v_+}(t+x/v_\up{f})=r_{v_+}(t)\exp[\lambda(t+x/v_\up{f})]$. Finally, we get the generalized eigenvalue problem
\begin{subequations}
\label{lineecirlin}
\begin{align}
&\mathcal{L}'_0\left\{r_{v_-},r_{v_+},\rvec_\yvec,\dot\rvec_\yvec+\lambda \rvec_\yvec\right\}=0 \\
&\mathcal{L}'_L\left\{r_{v_-}\espo{-\lambda\tau_\up{f}},r_{v_+}\espo{\lambda\tau_\up{f}},\rvec_\yvec,\dot\rvec_\yvec+\lambda \rvec_\yvec\right\}=0 \\
&\totder{\rvec_\yvec}{t}+\lambda\rvec_\yvec=\Amat(t)\rvec_\yvec
\end{align}
\end{subequations}
where $\Amat(t)$ is the Jacobian of $\qvec$ with respect to $\yvec$, and the linear operators $\mathcal{L}'_0$ and $\mathcal{L}'_L$ are derived from $\mathcal{L}_0$ and $\mathcal{L}_L$ exploiting the linear transformation \eqref{lineeqsol}.

\begin{figure}
\centerline{
\includegraphics[width=.6\columnwidth]{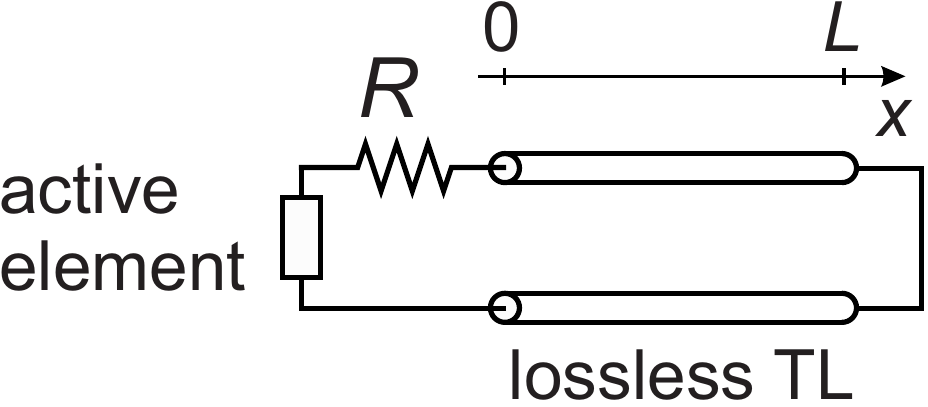}}
\caption{\label{cir}(Color online) Example of distributed resonating circuit containing a lossless TL. The active device is here represented as a constant resistance $R_\up{a}$ that may assume negative values.}
\end{figure}

As an example of application, we consider the distributed resonating circuit in Fig.~\ref{cir}: the active device is used to provide the necessary energy to overcome the dissipation included in resistance $R$. For the sake of simplicity, we assume here that the electrical equivalent of the active element is simply a constant resistance $R_\up{a}$, whose value however may become negative.

In this case, \eqref{lineecir3} disappears, and the two linear operators $\mathcal{L}_0$ and $\mathcal{L}_L$ simply read
\begin{subequations}
\label{lineecirres}
\begin{align}
&v(0,t)+(R+R_\up{a})i(0,t)=0 \\
&v(L,t)=0.
\end{align}
\end{subequations}
Using $v_-$ and $v_+$, \eqref{lineecirres} provide
\begin{subequations}
\label{lineecirres1}
\begin{align}
&v_+(t)=\Gamma_0v_-(t) \\
&v_-(t)=-v_+(t-2\tau_\up{f})
\end{align}
\end{subequations}
where $\Gamma_0=[(R+R_\up{a})Y_0-1]/[(R+R_\up{a})Y_0+1]$ is the reflection coefficient \cite{Sorrentino} measured at the left of position $x=0$ and calculated using $1/Y_0$ as the reference impedance. Clearly, the previous conditions lead to
\begin{equation}
v_+(t)=-\Gamma_0v_+(t-2\tau_\up{f}).
\end{equation}
The eigenvalue problem \eqref{lineecirlin}, finally, is
\begin{equation}
\espo{2\lambda\tau_\up{f}}=-\Gamma_0.
\label{kyenge}
\end{equation}
therefore the bifurcation ($\real{\lambda}=0$) takes place for $\abs{\Gamma_0}=1$. Notice that for finite $R$, the only possible realization is $\Gamma_0=-1$, i.e., as expected, $R_\up{a}=-R$. In this case, the admitted oscillation frequencies are defined by the complex roots of
\begin{equation}
\espo{2\lambda\tau_\up{f}}=1
\end{equation}
i.e., $\lambda=\gei 2k\pi f$ where $f= 1/(2\tau_\up{f})$ and $k\in\Z$.

\section{Conclusion}

In this paper we have proved an important extension of the recently developed generalized Floquet theory~\cite{PRL} to systems supporting infinite memory. In particular, we have proved that a lower asymptotic bound exists for the
Floquet exponents of such cases. We have then analyzed three cases of systems with memory: an ideal 1D system, a Brownian particle, and a circuit resonator with an ideal transmission line. Owing to the fact that dynamical systems with memory are ubiquitous in science and technology, we expect our generalized Floquet theory will find numerous applications in diverse fields. 

\section{Acknowledgments} This work has been partially supported by the Spanish
Project TEC2011-14253-E, NSF grant No. DMR-0802830 and
the Center for Magnetic Recording Research at UCSD. The authors also acknowledge prof. Renato Orta of Politecnico di Torino for useful discussions on transmission line theory.

\appendix
\section{Proof of Lemma~\ref{lemma1}}
\label{Lemmaproof}

Let us consider a time increment $\delta s\ll \bar s$. The solution for $\bar s+\delta s$ is expressed as $\rvec(t)=\bar\rvec(t)+\delta\rvec(t)$ and $\lambda=\bar\lambda+\delta\lambda$. Since $\delta s$ is small, linearizing \eqref{finitemem1} we find the (first order) relationship between the Floquet quantities variations
\begin{equation}
\pvec(t)\delta\lambda+\delta\vvec(t)=\int_{t-\bar s-\delta s}^{t-\bar s} \Kmat(t,\tau)\rvec(\tau)\espo{\bar\lambda(\tau-t)}~\D\tau
\label{AAA}
\end{equation}
where we have defined the $T$ periodic functions
\begin{align}
\pvec(t)&= \rvec(t)-\int_{t-\bar s}^t \Kmat(t,\tau)\bar\rvec(\tau)(\tau-t)\espo{\bar\lambda(\tau-t)}~\D\tau, \label{kuwa}\\
\delta\vvec(t)=&\totder{\delta\rvec}{t}+\bar\lambda\delta\rvec(t)-\Amat(t)\delta\rvec(t)\nonumber\\
& - \int_{t-\bar s}^t \Kmat(t,\tau)\delta\rvec(\tau)\espo{\bar\lambda(\tau-t)}~\D\tau.
\end{align}
Since $\delta s$ is small, we can approximate the integral in \eqref{AAA} as
\begin{align}
& \int_{t-\bar s-\delta s}^{t-\bar s} \Kmat(t,\tau)\rvec(\tau)\espo{\bar\lambda(\tau-t)}~\D\tau \nonumber\\
&\quad\approx \Kmat(t,t-\bar s-\delta s/2)\rvec(t-\bar s-\delta s/2)\espo{-\bar\lambda(\bar s+\delta s/2)}.
\label{AAA1}
\end{align}

We now define the scalar product between $T$-periodic functions
\begin{equation}
(\avec(t),\bvec(t))=\fracd{1}{T}\int_0^T \avec^\dagger(t)\bvec(t)~\D t,
\end{equation}
where $^\dagger$ denotes hermitian conjugation, and we consider a versor $\evec(t)$ orthogonal to $\delta\vvec(t)$. Equations \eqref{AAA} and \eqref{AAA1} yield
\begin{align}
(\evec(t),\pvec(t))\delta\lambda\approx &\left(\evec(t),\Kmat(t,t-\bar s-\delta s/2)\rvec(t-\bar s-\delta s/2)\vphantom{\espo{-\bar\lambda(\bar s+\delta s/2)}}\right.\nonumber\\
&\left.\times\espo{-\bar\lambda(\bar s+\delta s/2)}\right)\delta s.
\label{AAAproj}
\end{align}
Defining $\abs{\alphavec}$ as the vector made of the collection of the absolute values of the components of $\alphavec(t)$, \eqref{AAAproj} implies
\begin{align}
&\abs{(\evec(t),\pvec(t))\delta\lambda}\le  \left|\left( \abs{\evec(t)}, \norm{\Kmat(t,t-\bar s-\delta s/2)} \right.\vphantom{\espo{-\bar\lambda(\bar s+\delta s/2)}}\right.\nonumber\\
&\qquad \times\left.\left.\abs{\rvec(t-\bar s-\delta s/2)} \right) \espo{-\bar\lambda(\bar s+\delta s/2)}\right|\nonumber\\
&\quad\le \left( \abs{\evec(t)}, \abs{\rvec(t-\bar s-\delta s/2)} \right) \espo{-\real{\bar\lambda}(\bar s+\delta s/2)}\nonumber\\
&\qquad \times \max_{t\in[0,T]} \norm{\Kmat(t,t-\bar s-\delta s/2)}
\label{AAAproj1}
\end{align}
therefore
\begin{align}
\abs{\delta\lambda}&\le \fracd{\left( \abs{\evec(t)}, \abs{\rvec(t-\bar s-\delta s/2)} \right)}{\abs{(\evec(t),\pvec(t))}}
\espo{-\real{\bar\lambda}(\bar s+\delta s/2)}\nonumber\\
&\qquad \times \delta s\max_{t\in[0,T]} \norm{\Kmat(t,t-\bar s-\delta s/2)}.
\label{AAAproj2}
\end{align}

Since $\alpha\exp(\real{\bar\lambda}\alpha)$ is a monotonic function of $\alpha$, $\bar\tau\in[t-\bar s,t]$ exists such that
\begin{align}
&\int_{t-\bar s}^t \Kmat(t,\tau)\bar\rvec(\tau)(\tau-t)\espo{\bar\lambda(\tau-t)}~\D\tau \nonumber\\
&\qquad = \Kmat(t,\bar\tau)\bar\rvec(\bar\tau)\espo{\gei\imag{\bar\lambda}(\bar\tau-t)}\nonumber\\
&\qquad\quad\times \int_{t-\bar s}^t (\tau-t)\espo{\real{\bar\lambda}(\tau-t)}~\D\tau\nonumber\\
&\qquad = \Kmat(t,\bar\tau)\bar\rvec(\bar\tau)\espo{\gei\imag{\bar\lambda}(\bar\tau-t)}\espo{-\real{\bar\lambda}\bar s}\nonumber\\
&\qquad\quad\times \left[ \fracd{\bar s}{\real{\bar\lambda}} + \fracd{1}{\real{\bar\lambda}^2} - \fracd{\espo{\real{\bar\lambda}\bar s}}{\real{\bar\lambda}^2} \right].
\label{AAAB}
\end{align}
Defining $\pvec'(t)=\pvec(t)\exp\left(\real{\bar\lambda}\bar s\right)$, from \eqref{kuwa} and \eqref{AAAB} we find that for $\bar s\gg 0$
\begin{subequations}
\label{baba}
\begin{align}
& \pvec'(t)\approx -\fracd{\Kmat(t,\bar\tau)\bar\rvec(\bar\tau)\espo{\gei\imag{\bar\lambda}(\bar\tau-t)}}{\real{\bar\lambda}}\bar s
 \nonumber\\
&\qquad\qquad \qquad\qquad \qquad \qquad \qquad  \up{if $\real{\bar\lambda}<0$}\\
&
\pvec'(t)\approx\left[ \bar\rvec(t)+\fracd{\Kmat(t,\bar\tau)\bar\rvec(\bar\tau)\espo{\gei\imag{\bar\lambda}(\bar\tau-t)}}{\real{\bar\lambda}^2}\right]\espo{\real{\bar\lambda}\bar s}
 \nonumber\\
&\qquad\qquad \qquad\qquad \qquad \qquad \qquad  \up{if $\real{\bar\lambda}>0$}.
\end{align}
\end{subequations}
Since $\Kmat(t,\tau)$ satisifies \eqref{cond1}, $0<M<+\infty$ exists such that
\begin{align}
&\fracd{\abs{(\abs{\evec},\abs{\bar\rvec(t-\bar s-\delta s/2)})}}{\abs{(\evec,\pvec)}}=\fracd{\abs{(\abs{\evec},\abs{\bar\rvec(t-\bar s-\delta s/2)})}}{\abs{(\evec,\pvec')}}\nonumber\\
&\quad\times
\espo{\real{\bar\lambda}\bar s}
\le M \espo{\real{\bar\lambda}\bar s}.
\end{align}
Therefore, because of \eqref{baba}, for $\bar s\gg 0$, $0<M_+,M_-\le M$ exist such that
\begin{subequations}
\label{mama}
\begin{align}
& \fracd{\abs{(\abs{\evec},\abs{\bar\rvec(t-\bar s-\delta s/2)})}}{\abs{(\evec,\pvec')}}\le\fracd{M_-}{\bar s}
 \nonumber\\
&\qquad\qquad \qquad\qquad \qquad \qquad \qquad  \up{if $\real{\bar\lambda}<0$}\\
&
\fracd{\abs{(\abs{\evec},\abs{\bar\rvec(t-\bar s-\delta s/2)})}}{\abs{(\evec,\pvec')}}\le M_+\espo{-\real{\bar\lambda}\bar s}
 \nonumber\\
&\qquad\qquad \qquad\qquad \qquad \qquad \qquad  \up{if $\real{\bar\lambda}>0$}.
\end{align}
\end{subequations}
Accordingly, to first order in $\delta s$ \eqref{AAAproj2} becomes
\begin{subequations}
\label{QQQ}
\begin{align}
\abs{\delta\lambda}&\le M \delta s \max_{t\in[0,T]} \norm{\Kmat(t,t-\bar s-\delta s/2)}
\label{AAAproj3}
\end{align}
and, for $\bar s\gg0$
\label{bibi}
\begin{align}
&
\abs{\delta\lambda}\le \fracd{M_-}{\bar s} \delta s \max_{t\in[0,T]} \norm{\Kmat(t,t-\bar s-\delta s/2)}
 \nonumber\\
&\qquad\qquad \qquad\qquad \qquad \qquad \qquad  \up{if $\real{\bar\lambda}<0$}\\
&
\abs{\delta\lambda}\le M_+ \espo{-\real{\bar\lambda}\bar s} \delta s \max_{t\in[0,T]} \norm{\Kmat(t,t-\bar s-\delta s/2)}
 \nonumber\\
&\qquad\qquad \qquad\qquad \qquad \qquad \qquad  \up{if $\real{\bar\lambda}>0$}.
\end{align}
\end{subequations}

Let us now consider any $s>\bar s$. We divide the interval $[\bar s,s]$ into $N$ sub-intervals of size $\delta s=(s-\bar s)/N$, and denote as $\delta\lambda_j$ the Floquet exponent variation due to the $j$-th interval (with respect to the value attained at the beginning of the interval itself), such that
\begin{equation}
\abs{\lambda-\lambda_j}=\abs{\sum_{j=1}^N\delta\lambda_j}\le\sum_{j=1}^N\abs{\delta\lambda_j}.
\end{equation}
Taking the limit for $N\rightarrow +\infty$ and using \eqref{QQQ}, we find \eqref{conditions}.

\section{Proof of Theorem~\ref{theo1}}
\label{theoproof}

Exploiting a procedure akin to that in Appendix~\ref{Lemmaproof}, it is easy to prove that both $\delta \vvec(t)$ and  $\delta\rvec(t)$ tend to zero for large $\bar s$. In particular, for $\bar s\gg 0$ and $s>\bar s$
\begin{subequations}
\begin{align}
&
\norm{\rvec(t)-\bar\rvec(t)}\le H_-\int_{t-s}^{t-\bar s} \max_{t\in[0,T]} \norm{\Kmat(t,\tau)}~\D\tau
 \nonumber\\
&\qquad\qquad \qquad\qquad \qquad \qquad \qquad  \up{if $\real{\bar\lambda}<0$}\\
&
\norm{\rvec(t)-\bar\rvec(t)}\le H_+\int_{t-s}^{t-\bar s} \espo{-\real{\bar\lambda}(t-\tau)} \nonumber\\
&\quad \qquad \qquad \qquad \qquad \qquad \times\max_{t\in[0,T]} \norm{\Kmat(t,\tau)}~\D\tau
 \nonumber\\
&\qquad\qquad \qquad\qquad \qquad \qquad \qquad  \up{if $\real{\bar\lambda}>0$}.
\end{align}
\end{subequations}
where $H_-, H_+$ are positive constants. This result, together with Lemma~\ref{lemma1}, implies that for $\bar s\gg0$ and for any $s>\bar s$ the solution of
\begin{equation}
\totder{\rvec}{t}+\lambda\rvec(t)=\Amat(t)\rvec(t)+\int_{t-s}^t \Kmat(t,\tau)\rvec(\tau)\espo{\lambda(\tau-t)}~\D\tau
\end{equation}
becomes independent of $s$. Since this may happen only if the integral is independent of $s$, the latter should converge even if $\real{\lambda}<0$. Defining the critical exponent as in \eqref{tembo}, the integral may converge independently of $s$ if \eqref{simba1} is met.

\section{The Harmonic Balance approach}
\label{AHB}

Harmonic Balance (HB) is a powerful numerical technique used to transform differential equations into algebraic systems that can be applied when the terms and the solution in the differential equation are time-periodic. As such, it is widely used in circuit analysis and design tools, see, e.g., \cite{Sangiovanni}. In other words, HB seeks directly for the time-periodic solution without any explicit time-domain integration. Thus the transient part of the solution is avoided altogether.

Consider first a scalar, real function $\alpha(t)$, whose frequency domain representation is built by means of the (truncated) exponential Fourier series
\begin{equation}
\alpha(t)=\sum_{h=-N_\up{H}}^{N_\up{H}} \tilde{\alpha}_{h}\espo{\gei h\omega_0t}
\label{fs}
\end{equation}
where $\tilde{\alpha}_{h}$ is the $h$-th harmonic amplitude associated to the (angular) frequency $h\omega_0=h2\pi/T$ ($h$-th harmonic). Since $\alpha(t)$ is real, $ \tilde{\alpha}_{-h}=\tilde{\alpha}_{h^\star}$ ($^\star$ denotes complex conjugation): therefore only $2N_\up{H}+1$ real coefficients fully define the Fourier series. For numerical implementation, \eqref{fs} is replaced by the more effective trigonometric series representation \cite{Sangiovanni}. However, we will stick here to the exponential form for what concerns theoretical developments.

The $]0,T]$ fundamental period is discretized in a set of $2N_\up{H}+1$ time samples $t_k$ ($k=1,\dots,2N_\up{H}+1$), and the collection of the time sampled variable  $\breve{\alphavec}=[\alpha(t_1),\alpha(t_2),\dots,\alpha(t_{2N_\up{H}+1})]^\up{T}$ is put in relation with the collection of harmonic amplitudes $\tilde{\alphavec}=[\tilde{\alpha}_{-N_\up{H}},\tilde{\alpha}_{-N_\up{H}+1},\dots,\tilde{\alpha}_{0},\dots,\tilde{\alpha}_{N_\up{H}}]^\up{T}$ by means of the discrete Fourier transform (DFT) invertible linear operator $\Gammamat^{-1}$
\begin{equation}
\breve{\alphavec}=\Gammamat^{-1}\tilde{\alphavec} \Longleftrightarrow \tilde{\alphavec}=\Gammamat\breve{\alphavec}.
\label{DFT}
\end{equation}
Clearly, for $N_\up{H}\rightarrow\infty$ $\Gammamat^{-1}$ is the matrix representation of the operator defining the Fourier series representation of a $T$-periodic function.

In the frequency domain the time derivative is represented by a diagonal matrix $\Omegamat\in\mathbb{C}^{(2N_\up{H}+1)\times(2N_\up{H}+1)}$ proportional to $\omega_0$ \cite{Sangiovanni}
\begin{equation}
\tilde{\dot{\alphavec}}=\Gammamat\breve{\dot{\alphavec}}=\Omegamat\tilde{\alphavec}
\label{der}
\end{equation}
where $\dot{\alpha}(t)=\D\alpha/\D t$.

For a vector variable $\alphavec(t)\in\mathbb{R}^n$, \eqref{DFT} and \eqref{der} are easily generalized by expanding each time sample $\alpha(t_i)$ into a vector $\alphavec(t_i)\in\mathbb{R}^n$, and therefore defining the collection $\breve{\alphavec}=[\alphavec^\up{T}(t_1),\alphavec^\up{T}(t_2),\dots,\alphavec^\up{T}(t_{2N_\up{H}+1})]^\up{T}\in\mathbb{R}^{n(2N_\up{H}+1)}$. Correspondingly, the HB representation is $\tilde{\alphavec}=[\tilde{\alphavec}^\up{T}_{-N_\up{H}},\dots,\tilde{\alphavec}^\up{T}_{0},\dots,\tilde{\alphavec}^\up{T}_{N_\up{H}}]^\up{T}\in\mathbb{C}^{n(2N_\up{H}+1)}$. This allows to formally maintain \eqref{DFT} and \eqref{der} by defining two block diagonal matrices $\Gammamat_n^{-1}$ and $\Omegamat_n$ built replicating $n$ times the fundamental operators $\Gammamat^{-1}$ and $\Omegamat$
\begin{equation}
\breve{\alphavec}=\Gammamat_n^{-1}\tilde{\alphavec} \qquad \tilde{\dot{\alphavec}}=\Omegamat_n\tilde{\alphavec}.
\end{equation}

The HB representation of $\betavec(t)=\Tmat(t)\alphavec(t)$ (i.e., the convolution in frequency domain) and of its time derivative, where $\Tmat(t)$ is a $T$-periodic matrix and $\alphavec(t)$ a $T$-periodic vector, is derived as follows. Denoting as $\breve{\Tmat}$ the $n\times n$ block diagonal matrix built expanding each element $t_{h,k}(t)$ of $\Tmat(t)$ as a $(2N_\up{H}+1)\times(2N_\up{H}+1)$ diagonal matrix formed by the time samples $\breve{\tvec}_{h,k}$, we have
\begin{equation}
\tilde{\betavec}=\tilde{\Tmat}\tilde{\alphavec} \qquad
\tilde{\dot{\betavec}}=\Omegamat_n\tilde{\betavec}=\Omegamat_n\tilde{\Tmat}\tilde{\alphavec}
\label{simba}
\end{equation}
where $\tilde{\Tmat}=\Gammamat_n\breve{\Tmat}\Gammamat_n^{-1}$ is  a Toepliz matrix  built assembling the Fourier coefficients of the elements of $\Tmat(t)$ \cite{Sangiovanni}.




\bibliography{genfloq}

\end{document}